\documentclass[aps, pra, onecolumn, 10pt, tightenlines, notitlepage, superscriptaddress, longbibliography]{revtex4-1}
\newcommand{\papertitle}{Statistical analysis of randomized benchmarking}

\usepackage{amsmath}
\usepackage{amsthm}
\usepackage{amssymb}
\usepackage{amsfonts}
\usepackage{mathrsfs}
\usepackage{pgfplots}
\usepackage{xparse}
\usepackage{graphicx}
\usepackage{color}
\usepackage{hyperref}
\usepackage{qcircuit}

\hypersetup{
	colorlinks,
	linkcolor={blue},
	citecolor={blue},
	urlcolor={blue}
}

\usepackage{cleveref}

\newcommand{\rv}{\rvert}
\newcommand{\lv}{\lvert}
\newcommand{\tV}{\vert\kern-0.25ex\vert\kern-0.25ex\vert}

\newcommand{\ct}{\ensuremath{^{\dagger}}}

\DeclareMathOperator{\tr}{Tr}

\newcommand{\cB}{\ensuremath{\mathcal{B}}}

\newcommand{\cN}{\ensuremath{\mathcal{N}}}

\newcommand{\bbE}{\ensuremath{\mathbb{E}}}

\newcommand{\bbS}{\ensuremath{\mathbb{S}}}

\newcommand{\bbV}{\ensuremath{\mathbb{V}}}

\newcommand{\bi}{\begin{itemize}}
\newcommand{\ei}{\end{itemize}}

\theoremstyle{plain}
\newtheorem{thm}{Theorem}
\newtheorem{lem}[thm]{Lemma}

\theoremstyle{definition}

\theoremstyle{remark}

\definecolor{nblue}{rgb}{0.2,0.2,0.7}
\definecolor{ngreen}{rgb}{0.1,0.5,0.1}
\definecolor{nred}{rgb}{0.8,0.2,0.2}
\definecolor{nblack}{rgb}{0,0,0}

\newcommand{\hide}[1]{}

\usepackage[T1]{fontenc}

\newcommand{\bb}{\mathbb}
\newcommand{\dd}{\ensuremath{\delta\!}}
\usepackage{bbm}
\newcommand{\unit}{\mathbbm{1}}

\newcommand{\subref}[2]{{\cref{#1}\hyperref[#1]{#2}}}

\begin{document}
\title{\papertitle}

\author{Robin Harper}
\affiliation{Centre for Engineered Quantum Systems, School of Physics, The
University of Sydney, Sydney, Australia}
\author{Ian Hincks}
\affiliation{Institute for Quantum Computing and Department of Applied
Mathematics, University of Waterloo, Waterloo, Ontario N2L 3G1, Canada}
\affiliation{Quantum Benchmark Inc., Kitchener, ON N2H4C3, Canada}
\author{Chris Ferrie}
\affiliation{Centre for Quantum Software and Information, University of Technology Sydney, Australia}
\author{Steven T. Flammia}
\affiliation{Centre for Engineered Quantum Systems, School of Physics, The
University of Sydney, Sydney, Australia}
\affiliation{Yale Quantum Institute, Yale University, New Haven, CT 06520, USA}
\author{Joel J. Wallman}
\affiliation{Institute for Quantum Computing and Department of Applied
Mathematics, University of Waterloo, Waterloo, Ontario N2L 3G1, Canada}

\date{\today}

\begin{abstract}
Randomized benchmarking and variants thereof, which we collectively call RB+, are widely used to characterize the performance of quantum computers because they are simple, scalable, and robust to state-preparation and measurement errors.
However, experimental implementations of RB+ allocate resources suboptimally and make ad-hoc assumptions that undermine the reliability of the data analysis.
In this paper, we propose a simple modification of RB+ which rigorously eliminates a nuisance parameter and simplifies the experimental design.
We then show that, with this modification and specific experimental choices, RB+ efficiently provides estimates of error rates with multiplicative precision.
Finally, we provide a simplified rigorous method for obtaining credible regions for parameters of interest and a heuristic approximation for these intervals that performs well in currently relevant regimes.
\end{abstract}

\maketitle

\section{Introduction}

Characterizing large scale quantum devices is a prerequisite to optimizing their performance and being able to reliably perform useful information processing tasks.
Full characterization is manifestly not scalable for general errors, so that
scalable methods can only partially characterize the noise.
Currently, the only fully scalable protocols that partially characterize quantum devices are randomized benchmarking~\cite{Knill2008, Magesan2011,Magesan2012a} and variants thereof (RB+)~\cite{Magesan2012, Carignan-Dugas2015, Wallman2015, Wallman2015b, Sheldon2016, Cross2016, Harper2017}.
This family of protocols can provide a wide variety of information about noise parameters, including the average error rate \cite{Knill2008, Magesan2011, Magesan2012a, Wallman2017}, error rates for specific gates \cite{Magesan2012, Carignan-Dugas2015, Cross2016, Harper2017}, leakage rates \cite{Wallman2016}, loss rates \cite{Wallman2015b, Wallman2017}, and the amount of residual unitary (calibration) errors \cite{Wallman2015, Sheldon2016,Yang2018}.

RB+ provides estimates of noise parameters by applying long sequences of random gates to amplify errors in the implementation of gates and estimate them independently from state preparation and measurement errors (SPAM).
Typically, descriptions of RB+ state that experiments should be repeated to obtain a desired precision without necessarily specifying (or recommending) any of the following: (1) estimators for finite data; (2) how many repetitions should be performed; or (3) how finite and heteroscedastic data should be fit to a specified model.
This last point is important because RB+ data are generally heteroscedastic, meaning that the variance across the data is non-uniform, since the variance over random sequences increases with the sequence length~\cite{Wallman2014}.
Abstaining from specifics on these points was perhaps warranted by the fact that particular choices are difficult to derive or justify as being optimal, or nearly optimal.
Obtaining a fully general and optimal specification is confounded by the unknown distribution of errors over the random sequence of gates~\cite{Wallman2014, Helsen2017}.
However, in one experimental regime, Bayesian techniques can be applied to obtain rigorous credible intervals for the model parameters as well as efficient allocation of experimental measurements~\cite{Granade2014, Granade2017}.
We discuss and utilize this work by Granade et al.\ in \cref{sec:arb}.

In this paper, we present a minor modification of RB+ that improves the efficiency of the method by eliminating a nuisance model parameter, yet it adds no experimental overhead.
Similar methods have been presented previously in the literature for the case of single-qubit RB~\cite{Muhonen2015, Fogarty2015}.
For the two remaining model parameters, we then provide estimators which do not have to be weighted to correct for heteroscedasticity because they can be estimated from two independent sequence lengths.
Without our modification, at least 3 sequence lengths are required, which in turn require a weighted fit where the correct weights are not generally inferable from the data.
We then study the distribution of the parameter estimators and show how to obtain simple and rigorous credible intervals in the regime studied in Ref.~\cite{Granade2014}, that is, when each random sequence is repeated once.
Finally, we also provide a simple proof that certain experimental design choices enable RB+ to efficiently provide estimates of error rates that have multiplicative precision.
By showing that the estimates of such error rates have multiplicative precision, we confirm that RB+ will continue to allow efficient estimation of the model parameters as gate fidelity rates improve through the simple expedient of increased sequence lengths.

In what follows, we use the notation that $\hat{x}$ is an estimator of a quantity $\bar{x}$, where the bar denotes that either an expected value or a sample average has been taken over realizations of a random variable $x$.

\section{RB+ protocol}\label{sec:RBProtocol}

We begin by providing a general framework that describes all existing RB+ protocols except the leakage protocol of Ref.~\cite{Wood2018} and the unitarity protocol of Ref.~\cite{Wallman2015}.
We will also present a modified version of the unitarity protocol.
We exclude the leakage protocol of Ref.~\cite{Wood2018} because it has a more complicated fit model that is not robust to SPAM errors.

RB+ protocols are of the following form.

\vspace{6pt}
\fbox{\parbox{\textwidth}{\parbox{0.95\textwidth}{
\begin{enumerate}
\item Choose a positive integer $m$.

\item Choose a random sequence of gates $s$ from a set $\bbS_m$, typically of Clifford gates.
Note that these gates are often chosen to leave a state invariant.
However, as discussed below, uniformly choosing $s$ to either leave the state invariant or map it to an orthogonal state eliminates a nuisance model parameter.

\item Obtain an estimate $\hat{q}(m,s)$ of the expectation value $q(m,s)$ of an observable $E$ after preparing a state $\rho$ and applying the gates in $s$.
Typically $\rho$ should be close to an ideal computational basis state and $E$ should be close to a projector onto a pure state in the computational basis.

\item Repeat steps 2--3 $k_m$ times to obtain an estimate $\hat{q}(m)$ of $\bar{q}(m) = \lv\bbS_m\rv^{-1} \sum_{s\in\bbS_m} q(m,s)$.

\item Repeat steps 1--4 and fit to the model
\begin{align}\label{eq:decay_model}
\bar{q}(m) = A p^m + B
\end{align}
where $p$ is related to some parameter of interest (e.g., the average gate fidelity to the identity) and $A$ and $B$ are SPAM-dependent constants.
\end{enumerate}}}}

\vspace{6pt}

We assume throughout that $A\gg 0$, which holds in current regions of interest, as otherwise it is unclear how to efficiently gather useful statistics.

\subsection{Unitarity}

We now introduce a slight variant of the RB+ protocol to handle the special case of the unitarity protocol from Ref.~\cite{Wallman2015}.
The variant enables an independent estimate of the unitarity (which quantifies how coherent the errors are) and the leakage rate~\cite{Wallman2015b,Wallman2016}.
Note that the following protocol is not strictly scalable as it involves sampling every Pauli matrix and also assumes that there is no (or minimal) state-dependent loss.
The scalability could be improved by, for example, performing importance sampling of the Pauli matrices conditioned on the sequence~\cite{Flammia2011}; however, we leave this as an open problem.

\begin{enumerate}
\item Choose a positive integer $m$.

\item Choose a random sequence of $m$ $n$-qubit Clifford gates $s$.

\item For each $n$-qubit Pauli matrix $P$, obtain an estimate $\hat{q}(m,s|P)$ of the expectation value of the observable $P$ after preparing a fixed state $\rho$ and applying the gates in $s$.

\item Repeat steps 2--3 $k_m$ times. For each $n$-qubit Pauli matrix $P\neq I$,
set
\begin{align}
\begin{split}
	\hat{a}(m|P) &= \sum_{s} \hat{q}(m,s|P)/k_m \\
	\hat{a}(m) &= \frac{1}{4^n - 1}\sum_P \hat{a}(m|P) \\
	\hat{b}(m) &= \frac{1}{k_m}\sum_{P,s} \hat{q}(m,s|P)^2 - \hat{a}(m|P)^2.
\end{split}
\end{align}

\item Repeat steps 1--4 and fit to the models
\begin{align}\label{eq:decay_model2}
\begin{split}
	\bar{a}(m) &= A l^m \\
	\bar{b}(m) &= A' u^m
\end{split}
\end{align}
where $l$ and $u$ are the leakage rate~\cite{Wallman2016} and the
unitarity~\cite{Wallman2015} respectively.
\end{enumerate}

Unlike other protocols, the combined unitarity/loss protocol does not require any truncation of $\bar{q}(m)$ to avoid negative values. We also note that  recently an alternative protocol has been proposed which proposes a method for efficient unitarity benchmarking in the regime of few-qubit Clifford gates~\cite{Dirkse2018}.

\subsection{Eliminating the offset}{\label{sec:elimating_offset}

While superficially benign, the variable offset $B$ in \cref{eq:decay_model} can severely increase the marginal uncertainty in $p$, the parameter of interest~\cite{Muhonen2015}.
We now present a method of rigorously and exactly eliminating this constant offset without having to estimate its value.

First note that for gate independent noise $\Lambda$, the constant
\begin{align}
B:=\tr\bigl[E\Lambda(\unit/2^n)\bigr]
\end{align}
and the decay parameter $p$~\cite{Magesan2012a,Wallman2017} do not change if we compile any gate into the sequence.
(In fact, strictly this holds even for non-trace-preserving noise where $B$ is multiplied by a second exponential).
In particular, let $X$ be any gate that maps the input state to an orthogonal state, as the single-qubit Pauli $X$ operator does for states in the computational basis.
Let $\bbS_{m,b}$ be the set of sequences obtained from $\bbS_m$ by compiling $X^b$ into the sequence and let
\begin{align}
	\hat{q}(m|b) = \frac{1}{\lv\bbS_{m,b}\rv}\sum_{s\in\bbS_{m,b}} q(m,s).
\end{align}
Then we have
\begin{align}
\label{eq:diffDecay}
	\bar{q}(m) = \bar{q}(m|0) - \bar{q}(m|1) = Ap^m
\end{align}
where now $A\in[0,1]$.
A similar idea was suggested for single qubits in \cite{Muhonen2015, Fogarty2015}.
One disadvantage of this approach is that the remaining $A$ coefficient in \cref{eq:decay_model} may be small for some values of $b$, especially for multiple qubits, thus reducing the signal from some experiments.

Alternatively, consider an $n$-qubit POVM $\{E_1,\ldots, E_k\}$ and suppose that a set of gates $\{X_1, \ldots, X_k\}$ are such that $E_j \approx X_j E_1 X_j\ct$.
Then by compiling $X_j$ into the sequence uniformly at random and recording the probability of observing the corresponding $E_j$ and averaging over $j$, the average value of $B$ becomes
\begin{align}
B = \frac{1}{k}\sum_{j=1}^k\tr\bigl[E_j\Lambda(\tfrac{\unit}{2^n})\bigr]
 = \frac{1}{k}\tr\bigl[\unit\Lambda(\tfrac{\unit}{2^n})\bigr]
 = \frac{1}{k}\,.
\end{align}

\section{Estimating the decay rate}

With a known value of $B$, \cref{eq:decay_model,eq:decay_model2} have two unknown parameters and so we need at least two values of $m$, denoted $m_1 < m_2$, to estimate either (or both) parameters.
Alternatively, we could use a single value of $m$ if we are content to accept a lower bound on $p$ (by assuming that $A = 1$, and that other model assumptions are respected), and this may be sufficient for certain purposes.

From \cref{eq:decay_model},
\begin{align}
\label{eq:LSestimator}
\begin{split}
A &= \left[\bar{q}(m_1)-B\right]^{m_2/\dd m} \left[\bar{q}(m_2)-B\right]^{-m_1/\dd m},  \\
p &= \left[\bar{q}(m_1)-B\right]^{-1/\dd m} \left[\bar{q}(m_2)-B\right]^{1/\dd m}.
\end{split}
\end{align}
Each of these terms is of the form $x_1^{\alpha_1} x_2^{\alpha_2}$ where
$x_j^\alpha = \left[\bar{q}(m_j)-B\right]^\alpha$.
A natural approach would be to estimate $x_j^\alpha$ by $\left[\hat{q}(m_j)-B\right]^\alpha$, where $\hat{q}(m_j)$ is an unbiased estimator for $\bar{q}(m_j)$ (e.g., the sample mean).
However, this approach has two issues.
First, the estimator is complex or undefined if $\hat{q}(m_j)\leq B$.
To address this issue, we can truncate $\bar{q}(m_j)$ to $B+\delta$ for some fixed $0<\delta\ll 1$, which will occur with negligible probability provided enough sequences are taken and the $\hat{q}(m_i)$ are sufficiently far from $0$.

A second issue is that if $\hat{q}(m_j)$ is an unbiased estimator for $\bar{q}(m_j)$, then $\left[\hat{q}(m_j)-B\right]^\alpha$ is a biased estimator of $x_j^\alpha$ for any $\alpha\neq 1$.
To address this second issue, we can estimate and then substract the bias if necessary.
Repeating steps 2--4 for a fixed $m$ and using final gates compiled in to remove the offset $B$ yields an estimate $\hat{q}(m_j) -B = x_j(1 + \epsilon_j)$ for some random variable $\epsilon_j$ with zero mean.
We then have
\begin{align}\label{eq:estimatorBias}
\bbE\left[\hat{q}(m_j)-B\right]^\alpha
&= x_j^\alpha\bbE(1 + \epsilon_j)^\alpha \notag\\
&= x_j^\alpha\left[1 + \tfrac{1}{2}\alpha(\alpha-1)\bbE \epsilon_j^2 + O(\bbE \alpha\epsilon^3)\right] \notag\\
&\approx x_j^\alpha + \tfrac{1}{2}\alpha(\alpha-1)x_j^{\alpha-2}\bbV \hat{q}(m_j).
\end{align}
Therefore $\left[\hat{q}(m_j)-B\right]^\alpha$ is biased but consistent and so the bias can be neglected when sufficiently many sequences are sampled at each $m_j$.
Moreover, the bias is modulated by $\alpha$, which as we prove below in \cref{sec:rigproof}, is $O(r)$ in the optimal regime for $p$.
For small numbers of sampled sequences, the bias term can be subtracted using sample estimates of $x_j^\alpha$ and $\bbV\hat{q}(m_j)$ on the right-hand-side of \cref{eq:estimatorBias}.
As our numerics will show in \cref{sec:arb}, the bias is negligible for intermediate numbers of measurements but is noticeable for very small numbers of sequences.

To determine approximately optimal values of $m_1$ and $m_2$ assuming that the bias in \cref{eq:estimatorBias} is negligible, note that
\begin{align}\label{eq:epsilonVar}
\bbV\left[\hat{q}(m_j)-B\right]^\alpha
&= x_j^{2\alpha-2} \alpha^2 \bbV[\hat{q}(m_j)] + O(\bbE \alpha^3\epsilon^3).
\end{align}

Now we note that by Chebyshev's inequality,
\begin{align}
	\Pr\Bigl(\lv \hat{p} - p\rv > k \sqrt{\bb V(\hat{p})}\Bigr) \leq k^{-2},
\end{align}
so that with probability 8/9 (for example) and using the trivial bound $p \le 1$, we have
\begin{align}\label{eq:estimatorError}
\lv\hat{p} - p\rv \leq \frac{3}{\dd m}\biggl(\sum_j \bbV[\hat{q}(m_j)] \biggr)^{1/2} + O\bigl(\bbE(\epsilon_j^{3/2})/\dd m^{3/2}\bigr).
\end{align}
The only unknowns in \cref{eq:estimatorError} are the variances at the two
sequence lengths.
Choosing $\dd m \approx 1/(1-p)$ therefore gives a multiplicative precision
estimate of the error rate $1-p$, as claimed.

We would like to make the error term as small as possible.
To achieve a large $\dd m$, and hence a small error, we want $m_1$ as small as possible and $m_2$ as large as possible.
However, \cref{eq:decay_model} is only typically accurate for $m\geq 4$~\cite{Wallman2017}, so we henceforth set $m_1=4$.
Furthermore, the number of sequences required to make the truncation probability negligible increases with $m_2$, and \cref{eq:epsilonVar} is inversely proportional to $x_j$ for $\alpha\leq 1$, so $m_2$ can only be increased to some fixed value.
Thus, when sampling to some fixed constant accuracy we must not choose $m_2$ to be too large.

Ref.~\cite{Granade2014} recommended $m_2 = \lceil 1/(1-p) \rceil$ as the optimal choice of $m_2$, however, this was for $B=0$ (i.e., the infinite-dimensional limit).
Numerically, we observe that $m_2 =\lceil 1/[2(1-p)] \rceil$ results in a more precise estimate.

Consequently, provided gate lengths ($m$) can be increased as specified above, the number of sequences required to determine $1-p$ to within a specified factor remain approximately independent of $p$.
This demonstrates that RB+ protocols scale favorably with the error rate.

This derivation is not quite rigorous only for a very trivial reason, namely because our use of Taylor's theorem requires that we control the smoothness of the functions being expanded over some region, and this region should also be appropriately defined.
These expressions are nonetheless useful for practical analysis of RB+.
By contrast, the derivation in \cref{sec:rigproof} is completely rigorous, but the proof is not meant to provide anything more than coarse guidance about how to choose parameter settings in practical situations.

\section{Accelerated RB}\label{sec:arb}

We now analyze accelerated randomized benchmarking (ARB)~\cite{Granade2014} using the modified protocol discussed in \cref{sec:elimating_offset}.
In ARB, each individual estimate $\hat{q}(m,s)$ is in $\{0,1\}$, that is, each random sequence is measured once.
Therefore, $\hat{q}(m)\sim \cB(k_m, \bar{q}(m))/k_m$, where $\cB(n,p)$ denotes the binomial distribution with $n$ trials and probability $p$.

For sufficiently many samples, log ratios of binomial variables are approximately normally distributed~\cite{Katz1978}, so that
\begin{align}
\begin{split}
\log\tfrac{\hat{q}(m_2)}{\hat{q}(m_1)} &\sim \cN\left(\log(\tfrac{\bar{q}(m_2)}{\bar{q}(m_1)}), \sigma^2\right)\\
\sigma^2 &= \sum_j\frac{\overline{q}(m_j)\left(1-\overline{q}(m_j)\right)}{k_j(\overline{q}(m_j)-B)^2}.
\end{split}
\end{align}
Therefore
\begin{align}\label{eq:normalARB}
\log\hat{p} = \frac{1}{\dd m}\log\tfrac{\hat{q}(m_2)}{\hat{q}(m_1)} \sim  \cN(\log p, \sigma^2/\dd m^2).
\end{align}

In \cref{fig:cdf}, we illustrate that the normal approximation is sufficiently accurate by comparing the exact cumulative density function for the estimator $\hat{p}$ from binomial statistics with the normal approximation of \cref{eq:normalARB} for multiple values of $A$, $B$, and $p$.
Note in particular that the shape of the cumulative density function for the estimator $\hat{p}$ is essentially independent of $p$, but has heavier tails for smaller values of $A$.

\begin{figure}
\includegraphics[width=0.9\linewidth]{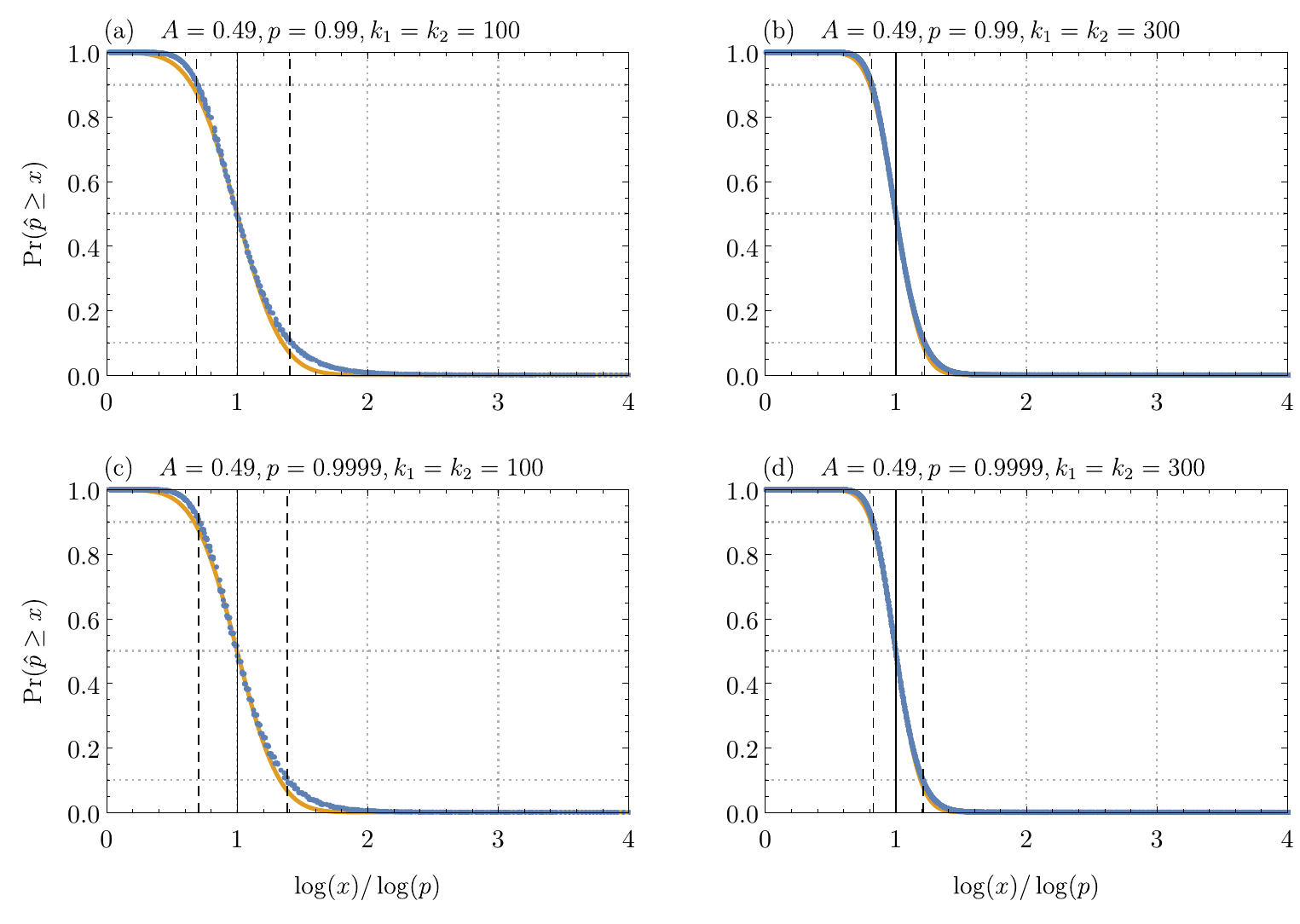}
\caption{Numerical demonstration that error rates can be accurately estimated with multiplicative precision using minimal resources. We plot the cumulative density function (CDF) for the estimator $\hat{p}$ from binomial statistics (blue dots) and \cref{eq:normalARB} (solid orange line) for different values of $A$, $p$ and $k$.
Note that the shape is essentially independent of $p$ but that smaller values of $A$ (that is, larger state-preparation and measurement errors) result in heavier tails.
Vertical dashed lines are located at the 10\% and 90\% quantiles of the CDF of the binomial distribution.}
\label{fig:cdf}
\end{figure}

Under the log-normal approximation for $\hat{p}$, the value of $m_2$ that minimizes the variance of the estimate is given by
\begin{align}
	\underset{m_2}{\operatorname{argmin}}\left(
	\log \left[
		p^{-2 m_1} \overline{q}(m_1) \left(1-\overline{q}(m_1)\right)
		+ p^{-2 m_2} \overline{q}(m_2) \left(1-\overline{q}(m_2)\right)
	\right]
	-2 \log (m_2-m_1)
	\right).
	\label{eq:bestm2}
\end{align}
While this minimizes the variance of $\log \hat{p}$ rather than
$\hat{p}$, since $p$ near $1$, we have $\bbV(\log\hat{p})\approx\bbV(\hat{p})$.
This minimization can be performed numerically using an initial value of $-1/\log p$. In \cref{fig:choosingm2}, we plot the variance and optimal $m_2$ values in several relevant parameter regimes.
The optimal value of $m_2$ depends on the true values of $A$, $B$, and $p$.
In \subref{fig:choosingm2}{(a-b)} we see that when choosing a future experiment based on present knowledge with multiplicative uncertainty in $p$, it is best to err on the side of $m_2$ that is short with respect to the optimal value.

\begin{figure}
\includegraphics[width=0.9\linewidth]{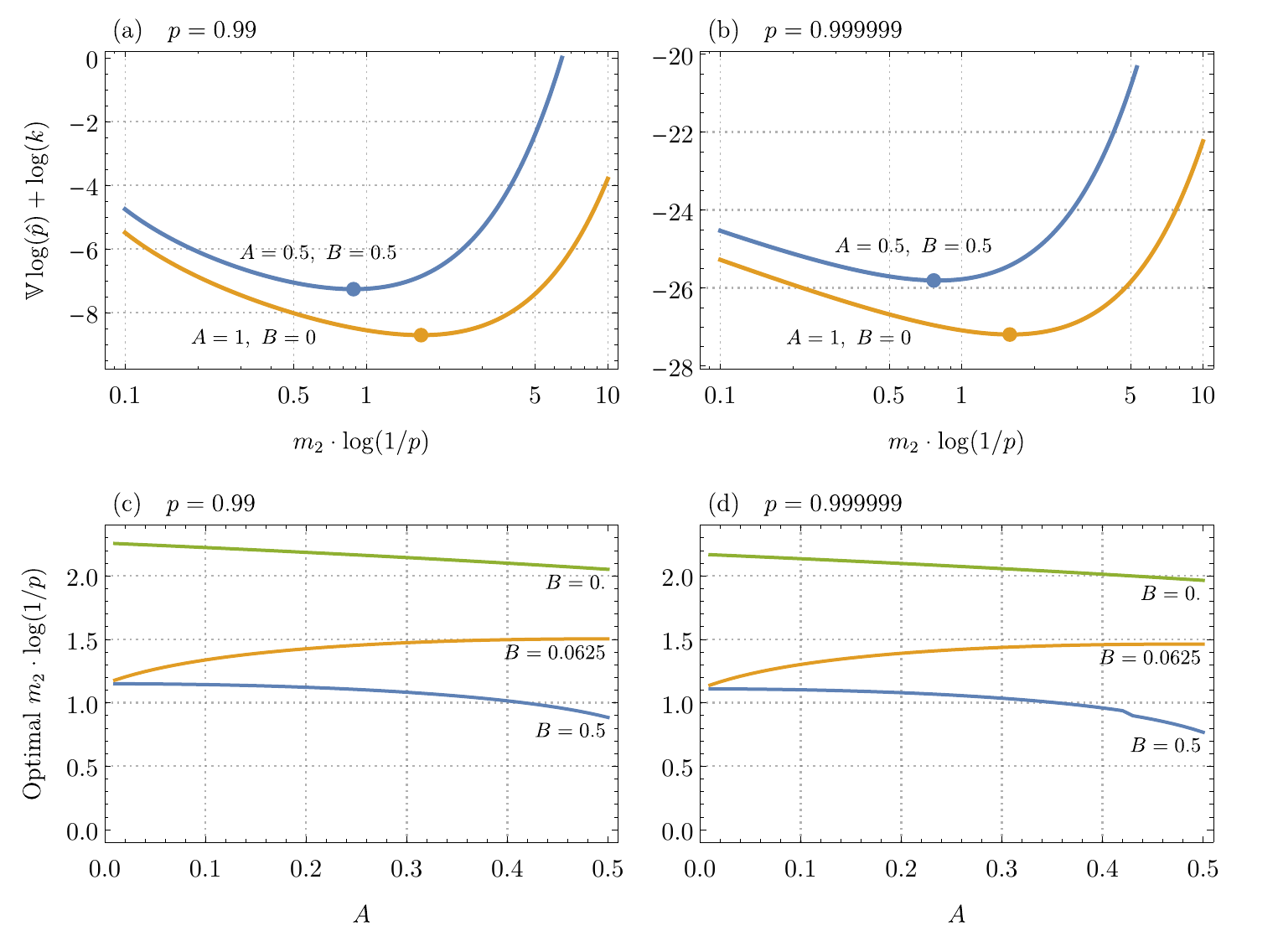}
\caption{Illustration that the variance only depends weakly on the choice of $m_2$ near the optimal value. Subfigures (a) and (b) show the variance of the estimated error rate in \cref{eq:normalARB} plotted for several combinations of $A$, $B$, and $p$, with points placed at the minimum of each curve.
Note that the $x$-axis is scaled logarithmically and such that the nominal value $m_2=1/\log(1/p)$ appears at the value $x=1$.
Subfigures (c) and (d) show the value of $m_2$ which minimizes the variance, \cref{eq:bestm2}, is plotted as a function of $A$ for several values of $B$ and $p$.
We see that the dependence on $p$ is essentially negligible for practical purposes, and that the dependence on $A$ is fairly weak.}
\label{fig:choosingm2}
\end{figure}

\section{Rigorous proof of multiplicative precision for the ratio estimator}\label{sec:rigproof}

In this section we give a rigorous proof that RB converges to an estimate with multiplicative precision using the ratio estimator in \cref{eq:LSestimator}.
Here the focus is not on obtaining tight answers, but on having a simple and clear statement of the scaling of the precision that can be achieved assuming the decay model in \cref{eq:diffDecay}, and achieving a given sample complexity.
We therefore largely neglect to track estimation errors closely, focusing instead on the simplest proof possible and a big-$O$ estimate of the resources required.

Let $r = 1-p$ and fix some small $1/16 > \epsilon_0>0$.
We are most interested in the regime where $r$ is small, or equivalently $p$ is close to 1. We assume that we can estimate the quantities $q_i = A p^{m_i}$ with an unbiased estimator $\hat{q}_i = q_i + Ap\epsilon_i$ where the estimation error $\epsilon_i$ is a random variable with zero mean.
Note that this is multiplicative precision for the case $m_i = 1$, but for larger values of $m_i$ we have just rescaled an additive precision by $A p$ for algebraic convenience.
If the estimator is the sample mean of $t$ Bernoulli random variables with mean $q_i$, then $\epsilon_i = O(1/\sqrt{t})$ with high probability.

Under these conditions, we can estimate $\hat{p}$ using the following algorithm.
\begin{description}
\item[Algorithm 1] Ratio estimator for exponential regression.
\begin{enumerate}
\item Set $i:=1$ and $m_1 := 1$.

\item Estimate $\hat{q}_1 := q_1 (1+ \epsilon_1)$ using $t$ samples.

\item While $\hat{q}_i > \frac{1}{3}\hat{q}_{1}$, Do
\begin{itemize}
	\item Set $i:=i+1$,
	\item Set $m_i := 2^i+1$,
	\item Estimate $\hat{q}_i := q_i + Ap\epsilon_i$ using $t$ samples.
\end{itemize}

\item Set $\ell := i$ and $m=2^\ell$.
\item Return $\hat{p} := \Bigl(\frac{\hat{q}_\ell}{\hat{q}_1}\Bigr)^{1/m}$ and $\hat{r} := 1-\hat{p}$.
\end{enumerate}
\end{description}

We now rigorously prove that the above algorithm returns an estimator with multiplicative precision of $\hat{r}$.

\begin{thm}
For any sufficiently small $\epsilon_0>0$, the algorithm above returns estimates $\hat{r}$ such that $|\hat{r} - r|\le O(\epsilon r)$ with probability $1-\delta$ using
\begin{align}
	M = O\biggl(\frac{1}{\epsilon^2}\log\bigl(\tfrac{1}{r}\bigr) \log\Bigl[\tfrac{1}{\delta}\log\bigl(\tfrac{1}{r}\bigr)\Bigr]\biggr)
\end{align}
measurements.
\end{thm}

The proof relies on a few simple lemmas, which we now state and prove.
\begin{lem}
\label{lemma:Chernoff}
Given a set of $\ell$ independent estimates $\hat{q}_i$ obtained from sampling each $t$ times as described above, the probability that $|\hat{q}_i-q_i| \ge A p \epsilon$ for any $i > 1$ or  $|\hat{q}_1-q_1| \ge q_1 \epsilon$ is at most $\delta$ if we choose $t = O\bigl(\frac{1}{\epsilon^2} \log \frac{\ell}{\delta}\bigr)$, where the implied constant depends on $A p$.
\end{lem}

\begin{proof}
The proof is an elementary application of the Chernoff bound and the union bound. We omit the details.
\end{proof}

Thus, we can assume that each of the random estimates $\epsilon_i$ satisfies $|\epsilon_i| \le \epsilon$ in the algorithm, and we will fail with probability at most $\delta$.
Next, we will see that the algorithm converges in a small number of steps $\ell$, and with $m$ taking a value that scales like $1/r$.

\begin{lem}
\label{lemma:pmbound}
With probability at least $1-\delta$, the above algorithm converges with
$\ell = \Theta\bigl(\log \tfrac{1}{r} \bigr)$ using
$O\bigl(\frac{\ell}{\epsilon^2} \log \frac{\ell}{\delta}\bigr)$ total samples, and with $m$ such that
\begin{align}
\label{eq:pmbound}
	\frac{(1-4\epsilon)^2}{9} < p^m \le \frac{1+4\epsilon}{3}.
\end{align}
\end{lem}

\begin{proof}
The algorithm exits the while loop when
\begin{align}
	\hat{q}_\ell
	= Ap(p^{2^\ell} + \epsilon_\ell) \le \tfrac{1}{3} \hat{q}_1
	= \frac{Ap}{3}\bigl(1 + \epsilon_1 \bigr) .
\end{align}
As the algorithm did not exit for $i = \ell -1$,
\begin{align}
	\hat{q}_{\ell-1}
	= Ap(p^{2^{\ell -1}} + \epsilon_{\ell-1}) > \tfrac{1}{3} \hat{q}_1
	= \frac{Ap}{3}\bigl(1 + \epsilon_1 \bigr) .
\end{align}
Squaring the latter inequality and then rearranging both to be in terms of $p^m$ with $m=2^\ell$, we have
\begin{align}
	\frac{(1-3\epsilon_{\ell-1}+\epsilon_1)^2}{9} < p^m \le \frac{1-3\epsilon_{\ell}+\epsilon_1}{3}.
\end{align}
Supposing that for some fixed $\epsilon>0$, $\lv\epsilon_i\rv \leq \epsilon$ for all $i$ with
with probability $1-\delta$, the claim about $p^m$ follows by taking the worst-case choices of the $\epsilon_i$.
Taking logarithms of this and using  $p=1-r$, we find that for any sufficiently small $\epsilon$ we have $\ell = \Theta\bigl[-\log\bigl(-\log (1-r)\bigr)\bigr]$.
As long as $r$ is bounded away from $1$ then $-\log(1-r) = \Theta(r)$, and this is equivalent to $\ell = \Theta\bigl(\log\tfrac{1}{r}\bigr)$.

If we sample as per \cref{lemma:Chernoff}, then $|\epsilon_i|\le \epsilon$ for all $i$ with
probability $1-\delta$. Therefore the claim about the total number of samples follows
immediately from \cref{lemma:Chernoff}.
\end{proof}

Now we are ready to prove the main theorem.
\begin{proof}[Proof of Theorem]
From the above lemmas, we know that with probability $1-\delta$ the ratio estimator converges with $m = O\bigl(\log \tfrac{1}{r}\bigr)$ and errors bounded by $\epsilon$ in the numerator and denominator.
We have the bounds $\hat{p}_- \le \hat{p} \le \hat{p}_+$, where
\begin{align}
\label{eq:phat}
	\hat{p}_\pm
    := \left(\frac{q_\ell\pm Ap\epsilon}{q_1\mp  Ap\epsilon}\right)^{1/m}
    = p  \left(\frac{1\pm\epsilon/p^m}{1\mp  \epsilon}\right)^{1/m}.
\end{align}
From the inequality in \cref{eq:pmbound}, we have
\begin{align}
	\hat{p}_+ \le p \left(1 + \frac{2 \epsilon  \left(5-4 \epsilon+8 \epsilon ^2\right)}{(1-4 \epsilon )^2 (1-\epsilon)}\right)^{1/m}
		\quad \text{and } \quad
	\hat{p}_- > p \left(1 - \frac{2 \epsilon  \left(5-4 \epsilon+8 \epsilon ^2\right)}{(1-4 \epsilon )^2 (1+\epsilon)}\right)^{1/m} .
\end{align}
Now we choose any $\epsilon_0 < 1/16$ so that the $\epsilon$ dependent terms above are $O(\epsilon)$ and the term for $\hat{p}_-$ remains less than 1.
Explicitly evaluating the $\epsilon$ dependent terms, we have the bounds
\begin{align}
	\hat{p}_+ < p \bigl(1 + O(\epsilon)\bigr)^{1/m}
		\quad \text{and } \quad
	\hat{p}_- > p \bigl(1 - O(\epsilon)\bigr)^{1/m} ,
\end{align}
where the implied constant decreases with $\epsilon_0$.
Now Taylor expanding in $1/m$ and using the result from \cref{lemma:pmbound} that $\ell = \log_2 m = \Theta\bigl(\log \tfrac{1}{r}\bigr)$, we find that
\begin{align}
	\hat{p}_+ < p\bigl(1+O(\epsilon q)\bigr) \quad \text{and} \quad \hat{p}_- > p\bigl(1-O(\epsilon r)\bigr).
\end{align}
Adopting the bounds $\hat{r}_\pm = 1-\hat{p}_\mp$ gives the analogous result for $\hat{r}$.
This establishes that the estimator has multiplicative precision,
\begin{align}
	|\hat{p} - p| = |\hat{r} - r|\le O(\epsilon r) \,.
\end{align}
The statement about complexity follows directly from the lemmas, and the theorem is proven.
\end{proof}

\section{Conclusion}

We have provided a modification to RB+ and concrete recommendations for how to obtain precise estimates of error rates in practical regimes.
We have rigorously shown that the precision is multiplicative, and our derivations and numerics demonstrate the utility of the heuristics that we use.
Our recommendations are based upon the assumption that the model in \cref{eq:decay_model} is correct.
For standard randomized benchmarking, there are only two factors that can cause a deviation from \cref{eq:decay_model} for sequence lengths $m\ge4$, namely, noise that is time-dependent or non-Markovian~\cite{Epstein2014,Wallman2014,Ball2016,Wallman2017}.
Both types of noise are ubiquitous in experiments and neither can be detected using only two sequence lengths.
A standard approach is to take data from more sequence lengths, perform a joint fit and then use the goodness-of-fit as an indicator for non-Markovian noise or drift.
However, fitting more sequence lengths is nontrivial as the data are heteroscedastic.
Furthermore, adding more sequence lengths does not significantly increase the quality of the error estimates when \cref{eq:decay_model} is correct, and so performing a joint fit provides little extra information and introduces correlations between model estimation and model validation.
We instead recommend fitting data using only two sequence lengths and then using hypothesis testing to determine if data taken at other sequence lengths are consistent with the hypothesis that the noise is static and Markovian.

\acknowledgments

RH and STF were supported by the Australian Research Council through the Centre of Excellence in Engineered Quantum Systems CE170100009.
This research was supported by the US Army Research Office through grant numbers W911NF-14-1-0098 and W911NF-14-1-0103.
IH and JJW gratefully acknowledge contributions from the Canada First Research Excellence Fund, Industry Canada, the Province of Ontario, and Quantum Benchmark Inc.

\bibliography{library}

\begin{thebibliography}{25}%
\makeatletter
\providecommand \@ifxundefined [1]{%
 \@ifx{#1\undefined}
}%
\providecommand \@ifnum [1]{%
 \ifnum #1\expandafter \@firstoftwo
 \else \expandafter \@secondoftwo
 \fi
}%
\providecommand \@ifx [1]{%
 \ifx #1\expandafter \@firstoftwo
 \else \expandafter \@secondoftwo
 \fi
}%
\providecommand \natexlab [1]{#1}%
\providecommand \enquote  [1]{``#1''}%
\providecommand \bibnamefont  [1]{#1}%
\providecommand \bibfnamefont [1]{#1}%
\providecommand \citenamefont [1]{#1}%
\providecommand \href@noop [0]{\@secondoftwo}%
\providecommand \href [0]{\begingroup \@sanitize@url \@href}%
\providecommand \@href[1]{\@@startlink{#1}\@@href}%
\providecommand \@@href[1]{\endgroup#1\@@endlink}%
\providecommand \@sanitize@url [0]{\catcode `\\12\catcode `\$12\catcode
  `\&12\catcode `\#12\catcode `\^12\catcode `\_12\catcode `\%12\relax}%
\providecommand \@@startlink[1]{}%
\providecommand \@@endlink[0]{}%
\providecommand \url  [0]{\begingroup\@sanitize@url \@url }%
\providecommand \@url [1]{\endgroup\@href {#1}{\urlprefix }}%
\providecommand \urlprefix  [0]{URL }%
\providecommand \Eprint [0]{\href }%
\providecommand \doibase [0]{http://dx.doi.org/}%
\providecommand \selectlanguage [0]{\@gobble}%
\providecommand \bibinfo  [0]{\@secondoftwo}%
\providecommand \bibfield  [0]{\@secondoftwo}%
\providecommand \translation [1]{[#1]}%
\providecommand \BibitemOpen [0]{}%
\providecommand \bibitemStop [0]{}%
\providecommand \bibitemNoStop [0]{.\EOS\space}%
\providecommand \EOS [0]{\spacefactor3000\relax}%
\providecommand \BibitemShut  [1]{\csname bibitem#1\endcsname}%
\let\auto@bib@innerbib\@empty
\bibitem [{\citenamefont {Knill}\ \emph {et~al.}(2008)\citenamefont {Knill},
  \citenamefont {Leibfried}, \citenamefont {Reichle}, \citenamefont {Britton},
  \citenamefont {Blakestad}, \citenamefont {Jost}, \citenamefont {Langer},
  \citenamefont {Ozeri}, \citenamefont {Seidelin},\ and\ \citenamefont
  {Wineland}}]{Knill2008}%
  \BibitemOpen
  \bibfield  {author} {\bibinfo {author} {\bibfnamefont {E.}~\bibnamefont
  {Knill}}, \bibinfo {author} {\bibfnamefont {D.}~\bibnamefont {Leibfried}},
  \bibinfo {author} {\bibfnamefont {R.}~\bibnamefont {Reichle}}, \bibinfo
  {author} {\bibfnamefont {J.}~\bibnamefont {Britton}}, \bibinfo {author}
  {\bibfnamefont {R.~B.}\ \bibnamefont {Blakestad}}, \bibinfo {author}
  {\bibfnamefont {J.~D.}\ \bibnamefont {Jost}}, \bibinfo {author}
  {\bibfnamefont {C.}~\bibnamefont {Langer}}, \bibinfo {author} {\bibfnamefont
  {R.}~\bibnamefont {Ozeri}}, \bibinfo {author} {\bibfnamefont
  {S.}~\bibnamefont {Seidelin}}, \ and\ \bibinfo {author} {\bibfnamefont
  {D.~J.}\ \bibnamefont {Wineland}},\ }\bibfield  {title} {\enquote {\bibinfo
  {title} {Randomized benchmarking of quantum gates},}\ }\href {\doibase
  10.1103/PhysRevA.77.012307} {\bibfield  {journal} {\bibinfo  {journal} {Phys.
  Rev. A}\ }\textbf {\bibinfo {volume} {77}},\ \bibinfo {pages} {012307}
  (\bibinfo {year} {2008})},\ \Eprint {http://arxiv.org/abs/0707.0963}
  {arXiv:0707.0963} \BibitemShut {NoStop}%
\bibitem [{\citenamefont {Magesan}\ \emph {et~al.}(2011)\citenamefont
  {Magesan}, \citenamefont {Gambetta},\ and\ \citenamefont
  {Emerson}}]{Magesan2011}%
  \BibitemOpen
  \bibfield  {author} {\bibinfo {author} {\bibfnamefont {Easwar}\ \bibnamefont
  {Magesan}}, \bibinfo {author} {\bibfnamefont {Jay~M.}\ \bibnamefont
  {Gambetta}}, \ and\ \bibinfo {author} {\bibfnamefont {Joseph}\ \bibnamefont
  {Emerson}},\ }\bibfield  {title} {\enquote {\bibinfo {title} {{Scalable and
  Robust Randomized Benchmarking of Quantum Processes}},}\ }\href {\doibase
  10.1103/PhysRevLett.106.180504} {\bibfield  {journal} {\bibinfo  {journal}
  {Physical Review Letters}\ }\textbf {\bibinfo {volume} {106}},\ \bibinfo
  {pages} {180504} (\bibinfo {year} {2011})},\ \Eprint
  {http://arxiv.org/abs/1009.3639} {arXiv:1009.3639} \BibitemShut {NoStop}%
\bibitem [{\citenamefont {Magesan}\ \emph
  {et~al.}(2012{\natexlab{a}})\citenamefont {Magesan}, \citenamefont
  {Gambetta},\ and\ \citenamefont {Emerson}}]{Magesan2012a}%
  \BibitemOpen
  \bibfield  {author} {\bibinfo {author} {\bibfnamefont {Easwar}\ \bibnamefont
  {Magesan}}, \bibinfo {author} {\bibfnamefont {Jay~M.}\ \bibnamefont
  {Gambetta}}, \ and\ \bibinfo {author} {\bibfnamefont {Joseph}\ \bibnamefont
  {Emerson}},\ }\bibfield  {title} {\enquote {\bibinfo {title} {{Characterizing
  quantum gates via randomized benchmarking}},}\ }\href {\doibase
  10.1103/PhysRevA.85.042311} {\bibfield  {journal} {\bibinfo  {journal}
  {Physical Review A}\ }\textbf {\bibinfo {volume} {85}},\ \bibinfo {pages}
  {042311} (\bibinfo {year} {2012}{\natexlab{a}})},\ \Eprint
  {http://arxiv.org/abs/1109.6887} {arXiv:1109.6887} \BibitemShut {NoStop}%
\bibitem [{\citenamefont {Magesan}\ \emph
  {et~al.}(2012{\natexlab{b}})\citenamefont {Magesan}, \citenamefont
  {Gambetta}, \citenamefont {Johnson}, \citenamefont {Ryan}, \citenamefont
  {Chow}, \citenamefont {Merkel}, \citenamefont {da~Silva}, \citenamefont
  {Keefe}, \citenamefont {Rothwell}, \citenamefont {Ohki}, \citenamefont
  {Ketchen},\ and\ \citenamefont {Steffen}}]{Magesan2012}%
  \BibitemOpen
  \bibfield  {author} {\bibinfo {author} {\bibfnamefont {Easwar}\ \bibnamefont
  {Magesan}}, \bibinfo {author} {\bibfnamefont {Jay~M.}\ \bibnamefont
  {Gambetta}}, \bibinfo {author} {\bibfnamefont {Blake~R.}\ \bibnamefont
  {Johnson}}, \bibinfo {author} {\bibfnamefont {Colm~A.}\ \bibnamefont {Ryan}},
  \bibinfo {author} {\bibfnamefont {Jerry~M.}\ \bibnamefont {Chow}}, \bibinfo
  {author} {\bibfnamefont {Seth~T.}\ \bibnamefont {Merkel}}, \bibinfo {author}
  {\bibfnamefont {Marcus~P.}\ \bibnamefont {da~Silva}}, \bibinfo {author}
  {\bibfnamefont {George~A.}\ \bibnamefont {Keefe}}, \bibinfo {author}
  {\bibfnamefont {Mary~B.}\ \bibnamefont {Rothwell}}, \bibinfo {author}
  {\bibfnamefont {Thomas~A.}\ \bibnamefont {Ohki}}, \bibinfo {author}
  {\bibfnamefont {Mark~B.}\ \bibnamefont {Ketchen}}, \ and\ \bibinfo {author}
  {\bibfnamefont {Matthias}\ \bibnamefont {Steffen}},\ }\bibfield  {title}
  {\enquote {\bibinfo {title} {{Efficient Measurement of Quantum Gate Error by
  Interleaved Randomized Benchmarking}},}\ }\href {\doibase
  10.1103/PhysRevLett.109.080505} {\bibfield  {journal} {\bibinfo  {journal}
  {Physical Review Letters}\ }\textbf {\bibinfo {volume} {109}},\ \bibinfo
  {pages} {080505} (\bibinfo {year} {2012}{\natexlab{b}})},\ \Eprint
  {http://arxiv.org/abs/arXiv:1203.4550} {arXiv:1203.4550} \BibitemShut
  {NoStop}%
\bibitem [{\citenamefont {Carignan-Dugas}\ \emph {et~al.}(2015)\citenamefont
  {Carignan-Dugas}, \citenamefont {Wallman},\ and\ \citenamefont
  {Emerson}}]{Carignan-Dugas2015}%
  \BibitemOpen
  \bibfield  {author} {\bibinfo {author} {\bibfnamefont {Arnaud}\ \bibnamefont
  {Carignan-Dugas}}, \bibinfo {author} {\bibfnamefont {Joel~J.}\ \bibnamefont
  {Wallman}}, \ and\ \bibinfo {author} {\bibfnamefont {Joseph}\ \bibnamefont
  {Emerson}},\ }\bibfield  {title} {\enquote {\bibinfo {title} {{Characterizing
  universal gate sets via dihedral benchmarking}},}\ }\href {\doibase
  10.1103/PhysRevA.92.060302} {\bibfield  {journal} {\bibinfo  {journal}
  {Physical Review A}\ }\textbf {\bibinfo {volume} {92}},\ \bibinfo {pages}
  {060302} (\bibinfo {year} {2015})},\ \Eprint
  {http://arxiv.org/abs/1508.06312} {arXiv:1508.06312} \BibitemShut {NoStop}%
\bibitem [{\citenamefont {Wallman}\ \emph
  {et~al.}(2015{\natexlab{a}})\citenamefont {Wallman}, \citenamefont {Granade},
  \citenamefont {Harper},\ and\ \citenamefont {Flammia}}]{Wallman2015}%
  \BibitemOpen
  \bibfield  {author} {\bibinfo {author} {\bibfnamefont {Joel~J.}\ \bibnamefont
  {Wallman}}, \bibinfo {author} {\bibfnamefont {Christopher}\ \bibnamefont
  {Granade}}, \bibinfo {author} {\bibfnamefont {Robin}\ \bibnamefont {Harper}},
  \ and\ \bibinfo {author} {\bibfnamefont {Steven~T.}\ \bibnamefont
  {Flammia}},\ }\bibfield  {title} {\enquote {\bibinfo {title} {{Estimating the
  Coherence of Noise}},}\ }\href {\doibase 10.1088/1367-2630/17/11/113020}
  {\bibfield  {journal} {\bibinfo  {journal} {New Journal of Physics}\ }\textbf
  {\bibinfo {volume} {17}},\ \bibinfo {pages} {113020} (\bibinfo {year}
  {2015}{\natexlab{a}})},\ \Eprint {http://arxiv.org/abs/1503.07865}
  {arXiv:1503.07865} \BibitemShut {NoStop}%
\bibitem [{\citenamefont {Wallman}\ \emph
  {et~al.}(2015{\natexlab{b}})\citenamefont {Wallman}, \citenamefont
  {Barnhill},\ and\ \citenamefont {Emerson}}]{Wallman2015b}%
  \BibitemOpen
  \bibfield  {author} {\bibinfo {author} {\bibfnamefont {Joel~J.}\ \bibnamefont
  {Wallman}}, \bibinfo {author} {\bibfnamefont {Marie}\ \bibnamefont
  {Barnhill}}, \ and\ \bibinfo {author} {\bibfnamefont {Joseph}\ \bibnamefont
  {Emerson}},\ }\bibfield  {title} {\enquote {\bibinfo {title} {{Robust
  Characterization of Loss Rates}},}\ }\href {\doibase
  10.1103/PhysRevLett.115.060501} {\bibfield  {journal} {\bibinfo  {journal}
  {Physical Review Letters}\ }\textbf {\bibinfo {volume} {115}},\ \bibinfo
  {pages} {060501} (\bibinfo {year} {2015}{\natexlab{b}})},\ \Eprint
  {http://arxiv.org/abs/1510.01272} {arXiv:1510.01272} \BibitemShut {NoStop}%
\bibitem [{\citenamefont {Sheldon}\ \emph {et~al.}(2016)\citenamefont
  {Sheldon}, \citenamefont {Bishop}, \citenamefont {Magesan}, \citenamefont
  {Filipp}, \citenamefont {Chow},\ and\ \citenamefont
  {Gambetta}}]{Sheldon2016}%
  \BibitemOpen
  \bibfield  {author} {\bibinfo {author} {\bibfnamefont {Sarah}\ \bibnamefont
  {Sheldon}}, \bibinfo {author} {\bibfnamefont {Lev~S.}\ \bibnamefont
  {Bishop}}, \bibinfo {author} {\bibfnamefont {Easwar}\ \bibnamefont
  {Magesan}}, \bibinfo {author} {\bibfnamefont {Stefan}\ \bibnamefont
  {Filipp}}, \bibinfo {author} {\bibfnamefont {Jerry~M.}\ \bibnamefont {Chow}},
  \ and\ \bibinfo {author} {\bibfnamefont {Jay~M.}\ \bibnamefont {Gambetta}},\
  }\bibfield  {title} {\enquote {\bibinfo {title} {{Characterizing errors on
  qubit operations via iterative randomized benchmarking}},}\ }\href {\doibase
  10.1103/PhysRevA.93.012301} {\bibfield  {journal} {\bibinfo  {journal}
  {Physical Review A}\ }\textbf {\bibinfo {volume} {93}},\ \bibinfo {pages}
  {012301} (\bibinfo {year} {2016})},\ \Eprint
  {http://arxiv.org/abs/1504.06597} {arXiv:1504.06597} \BibitemShut {NoStop}%
\bibitem [{\citenamefont {Cross}\ \emph {et~al.}(2016)\citenamefont {Cross},
  \citenamefont {Magesan}, \citenamefont {Bishop}, \citenamefont {Smolin},\
  and\ \citenamefont {Gambetta}}]{Cross2016}%
  \BibitemOpen
  \bibfield  {author} {\bibinfo {author} {\bibfnamefont {Andrew~W.}\
  \bibnamefont {Cross}}, \bibinfo {author} {\bibfnamefont {Easwar}\
  \bibnamefont {Magesan}}, \bibinfo {author} {\bibfnamefont {Lev~S.}\
  \bibnamefont {Bishop}}, \bibinfo {author} {\bibfnamefont {John~A.}\
  \bibnamefont {Smolin}}, \ and\ \bibinfo {author} {\bibfnamefont {Jay~M.}\
  \bibnamefont {Gambetta}},\ }\bibfield  {title} {\enquote {\bibinfo {title}
  {{Scalable randomised benchmarking of non-Clifford gates}},}\ }\href
  {\doibase 10.1038/npjqi.2016.12} {\bibfield  {journal} {\bibinfo  {journal}
  {npj Quantum Information}\ }\textbf {\bibinfo {volume} {2}},\ \bibinfo
  {pages} {16012} (\bibinfo {year} {2016})},\ \Eprint
  {http://arxiv.org/abs/arXiv:1510.02720} {arXiv:1510.02720} \BibitemShut
  {NoStop}%
\bibitem [{\citenamefont {Harper}\ and\ \citenamefont
  {Flammia}(2017)}]{Harper2017}%
  \BibitemOpen
  \bibfield  {author} {\bibinfo {author} {\bibfnamefont {Robin}\ \bibnamefont
  {Harper}}\ and\ \bibinfo {author} {\bibfnamefont {Steven~T.}\ \bibnamefont
  {Flammia}},\ }\bibfield  {title} {\enquote {\bibinfo {title} {{Estimating the
  fidelity of T gates using standard interleaved randomized benchmarking}},}\
  }\href {\doibase 10.1088/2058-9565/aa5f8d} {\bibfield  {journal} {\bibinfo
  {journal} {Quantum Science and Technology}\ }\textbf {\bibinfo {volume}
  {2}},\ \bibinfo {pages} {015008} (\bibinfo {year} {2017})},\ \Eprint
  {http://arxiv.org/abs/1608.02943} {arXiv:1608.02943} \BibitemShut {NoStop}%
\bibitem [{\citenamefont {{Wallman}}(2018)}]{Wallman2017}%
  \BibitemOpen
  \bibfield  {author} {\bibinfo {author} {\bibfnamefont {J.~J.}\ \bibnamefont
  {{Wallman}}},\ }\bibfield  {title} {\enquote {\bibinfo {title} {{Randomized
  benchmarking with gate-dependent noise}},}\ }\href {\doibase
  10.22331/q-2018-01-29-47} {\bibfield  {journal} {\bibinfo  {journal}
  {Quantum}\ }\textbf {\bibinfo {volume} {2}},\ \bibinfo {pages} {47} (\bibinfo
  {year} {2018})},\ \Eprint {http://arxiv.org/abs/1703.09835}
  {arXiv:1703.09835} \BibitemShut {NoStop}%
\bibitem [{\citenamefont {Wallman}\ \emph {et~al.}(2016)\citenamefont
  {Wallman}, \citenamefont {Barnhill},\ and\ \citenamefont
  {Emerson}}]{Wallman2016}%
  \BibitemOpen
  \bibfield  {author} {\bibinfo {author} {\bibfnamefont {Joel~J.}\ \bibnamefont
  {Wallman}}, \bibinfo {author} {\bibfnamefont {Marie}\ \bibnamefont
  {Barnhill}}, \ and\ \bibinfo {author} {\bibfnamefont {Joseph}\ \bibnamefont
  {Emerson}},\ }\bibfield  {title} {\enquote {\bibinfo {title} {{Robust
  characterization of leakage errors}},}\ }\href {\doibase
  10.1088/1367-2630/18/4/043021} {\bibfield  {journal} {\bibinfo  {journal}
  {New Journal of Physics}\ }\textbf {\bibinfo {volume} {18}},\ \bibinfo
  {pages} {043021} (\bibinfo {year} {2016})},\ \Eprint
  {http://arxiv.org/abs/arXiv:1412.4126} {arXiv:1412.4126} \BibitemShut
  {NoStop}%
\bibitem [{\citenamefont {{Yang}}\ \emph {et~al.}(2018)\citenamefont {{Yang}},
  \citenamefont {{Chan}}, \citenamefont {{Harper}}, \citenamefont {{Huang}},
  \citenamefont {{Evans}}, \citenamefont {{Hwang}}, \citenamefont {{Hensen}},
  \citenamefont {{Laucht}}, \citenamefont {{Tanttu}}, \citenamefont {{Hudson}},
  \citenamefont {{Flammia}}, \citenamefont {{Itoh}}, \citenamefont {{Morello}},
  \citenamefont {{Bartlett}},\ and\ \citenamefont {{Dzurak}}}]{Yang2018}%
  \BibitemOpen
  \bibfield  {author} {\bibinfo {author} {\bibfnamefont {C.~H.}\ \bibnamefont
  {{Yang}}}, \bibinfo {author} {\bibfnamefont {K.~W.}\ \bibnamefont {{Chan}}},
  \bibinfo {author} {\bibfnamefont {R.}~\bibnamefont {{Harper}}}, \bibinfo
  {author} {\bibfnamefont {W.}~\bibnamefont {{Huang}}}, \bibinfo {author}
  {\bibfnamefont {T.}~\bibnamefont {{Evans}}}, \bibinfo {author} {\bibfnamefont
  {J.~C.~C.}\ \bibnamefont {{Hwang}}}, \bibinfo {author} {\bibfnamefont
  {B.}~\bibnamefont {{Hensen}}}, \bibinfo {author} {\bibfnamefont
  {A.}~\bibnamefont {{Laucht}}}, \bibinfo {author} {\bibfnamefont
  {T.}~\bibnamefont {{Tanttu}}}, \bibinfo {author} {\bibfnamefont {F.~E.}\
  \bibnamefont {{Hudson}}}, \bibinfo {author} {\bibfnamefont {S.~T.}\
  \bibnamefont {{Flammia}}}, \bibinfo {author} {\bibfnamefont {K.~M.}\
  \bibnamefont {{Itoh}}}, \bibinfo {author} {\bibfnamefont {A.}~\bibnamefont
  {{Morello}}}, \bibinfo {author} {\bibfnamefont {S.~D.}\ \bibnamefont
  {{Bartlett}}}, \ and\ \bibinfo {author} {\bibfnamefont {A.~S.}\ \bibnamefont
  {{Dzurak}}},\ }\bibfield  {title} {\enquote {\bibinfo {title} {{Silicon qubit
  fidelities approaching stochastic noise limits via pulse optimisation}},}\
  }\href@noop {} {\bibfield  {journal} {\bibinfo  {journal} {ArXiv e-prints}\ }
  (\bibinfo {year} {2018})},\ \Eprint {http://arxiv.org/abs/1807.09500}
  {arXiv:1807.09500} \BibitemShut {NoStop}%
\bibitem [{\citenamefont {Wallman}\ and\ \citenamefont
  {Flammia}(2014)}]{Wallman2014}%
  \BibitemOpen
  \bibfield  {author} {\bibinfo {author} {\bibfnamefont {Joel~J.}\ \bibnamefont
  {Wallman}}\ and\ \bibinfo {author} {\bibfnamefont {Steven~T.}\ \bibnamefont
  {Flammia}},\ }\bibfield  {title} {\enquote {\bibinfo {title} {{Randomized
  benchmarking with confidence}},}\ }\href {\doibase
  10.1088/1367-2630/16/10/103032} {\bibfield  {journal} {\bibinfo  {journal}
  {New Journal of Physics}\ }\textbf {\bibinfo {volume} {16}},\ \bibinfo
  {pages} {103032} (\bibinfo {year} {2014})},\ \Eprint
  {http://arxiv.org/abs/arXiv:1404.6025} {arXiv:1404.6025} \BibitemShut
  {NoStop}%
\bibitem [{\citenamefont {Helsen}\ \emph {et~al.}(2017)\citenamefont {Helsen},
  \citenamefont {Wallman}, \citenamefont {Flammia},\ and\ \citenamefont
  {Wehner}}]{Helsen2017}%
  \BibitemOpen
  \bibfield  {author} {\bibinfo {author} {\bibfnamefont {Jonas}\ \bibnamefont
  {Helsen}}, \bibinfo {author} {\bibfnamefont {Joel~J.}\ \bibnamefont
  {Wallman}}, \bibinfo {author} {\bibfnamefont {Steven~T.}\ \bibnamefont
  {Flammia}}, \ and\ \bibinfo {author} {\bibfnamefont {Stephanie}\ \bibnamefont
  {Wehner}},\ }\bibfield  {title} {\enquote {\bibinfo {title} {{Multi-qubit
  Randomized Benchmarking Using Few Samples}},}\ }\href
  {http://arxiv.org/abs/1701.04299} {\  (\bibinfo {year} {2017})},\ \Eprint
  {http://arxiv.org/abs/1701.04299} {arXiv:1701.04299} \BibitemShut {NoStop}%
\bibitem [{\citenamefont {Granade}\ \emph {et~al.}(2014)\citenamefont
  {Granade}, \citenamefont {Ferrie},\ and\ \citenamefont {Cory}}]{Granade2014}%
  \BibitemOpen
  \bibfield  {author} {\bibinfo {author} {\bibfnamefont {Christopher}\
  \bibnamefont {Granade}}, \bibinfo {author} {\bibfnamefont {Christopher}\
  \bibnamefont {Ferrie}}, \ and\ \bibinfo {author} {\bibfnamefont {David~G.}\
  \bibnamefont {Cory}},\ }\bibfield  {title} {\enquote {\bibinfo {title}
  {{Accelerated Randomized Benchmarking}},}\ }\href {\doibase
  10.1088/1367-2630/17/1/013042} {\bibfield  {journal} {\bibinfo  {journal}
  {New Journal of Physics}\ }\textbf {\bibinfo {volume} {17}},\ \bibinfo
  {pages} {013042} (\bibinfo {year} {2014})},\ \Eprint
  {http://arxiv.org/abs/1404.5275v1} {arXiv:1404.5275v1} \BibitemShut {NoStop}%
\bibitem [{\citenamefont {Granade}\ \emph {et~al.}(2017)\citenamefont
  {Granade}, \citenamefont {Ferrie}, \citenamefont {Hincks}, \citenamefont
  {Casagrande}, \citenamefont {Alexander}, \citenamefont {Gross}, \citenamefont
  {Kononenko},\ and\ \citenamefont {Sanders}}]{Granade2017}%
  \BibitemOpen
  \bibfield  {author} {\bibinfo {author} {\bibfnamefont {Christopher}\
  \bibnamefont {Granade}}, \bibinfo {author} {\bibfnamefont {Christopher}\
  \bibnamefont {Ferrie}}, \bibinfo {author} {\bibfnamefont {Ian}\ \bibnamefont
  {Hincks}}, \bibinfo {author} {\bibfnamefont {Steven}\ \bibnamefont
  {Casagrande}}, \bibinfo {author} {\bibfnamefont {Thomas}\ \bibnamefont
  {Alexander}}, \bibinfo {author} {\bibfnamefont {Jonathan}\ \bibnamefont
  {Gross}}, \bibinfo {author} {\bibfnamefont {Michal}\ \bibnamefont
  {Kononenko}}, \ and\ \bibinfo {author} {\bibfnamefont {Yuval}\ \bibnamefont
  {Sanders}},\ }\bibfield  {title} {\enquote {\bibinfo {title} {{QInfer}:
  Statistical inference software for quantum applications},}\ }\href {\doibase
  10.22331/q-2017-04-25-5} {\bibfield  {journal} {\bibinfo  {journal}
  {Quantum}\ }\textbf {\bibinfo {volume} {1}},\ \bibinfo {pages} {5} (\bibinfo
  {year} {2017})},\ \Eprint {http://arxiv.org/abs/1610.00336}
  {arXiv:1610.00336} \BibitemShut {NoStop}%
\bibitem [{\citenamefont {Muhonen}\ \emph {et~al.}(2015)\citenamefont
  {Muhonen}, \citenamefont {Laucht}, \citenamefont {Simmons}, \citenamefont
  {Dehollain}, \citenamefont {Kalra}, \citenamefont {Hudson}, \citenamefont
  {Freer}, \citenamefont {Itoh}, \citenamefont {Jamieson}, \citenamefont
  {McCallum}, \citenamefont {Dzurak},\ and\ \citenamefont
  {Morello}}]{Muhonen2015}%
  \BibitemOpen
  \bibfield  {author} {\bibinfo {author} {\bibfnamefont {J~T}\ \bibnamefont
  {Muhonen}}, \bibinfo {author} {\bibfnamefont {A}~\bibnamefont {Laucht}},
  \bibinfo {author} {\bibfnamefont {S}~\bibnamefont {Simmons}}, \bibinfo
  {author} {\bibfnamefont {J~P}\ \bibnamefont {Dehollain}}, \bibinfo {author}
  {\bibfnamefont {R}~\bibnamefont {Kalra}}, \bibinfo {author} {\bibfnamefont
  {F~E}\ \bibnamefont {Hudson}}, \bibinfo {author} {\bibfnamefont
  {S}~\bibnamefont {Freer}}, \bibinfo {author} {\bibfnamefont {K~M}\
  \bibnamefont {Itoh}}, \bibinfo {author} {\bibfnamefont {D~N}\ \bibnamefont
  {Jamieson}}, \bibinfo {author} {\bibfnamefont {J~C}\ \bibnamefont
  {McCallum}}, \bibinfo {author} {\bibfnamefont {A~S}\ \bibnamefont {Dzurak}},
  \ and\ \bibinfo {author} {\bibfnamefont {A}~\bibnamefont {Morello}},\
  }\bibfield  {title} {\enquote {\bibinfo {title} {Quantifying the quantum gate
  fidelity of single-atom spin qubits in silicon by randomized benchmarking},}\
  }\href {http://stacks.iop.org/0953-8984/27/i=15/a=154205} {\bibfield
  {journal} {\bibinfo  {journal} {Journal of Physics: Condensed Matter}\
  }\textbf {\bibinfo {volume} {27}},\ \bibinfo {pages} {154205} (\bibinfo
  {year} {2015})},\ \Eprint {http://arxiv.org/abs/arXiv:1410.2338}
  {arXiv:1410.2338} \BibitemShut {NoStop}%
\bibitem [{\citenamefont {Fogarty}\ \emph {et~al.}(2015)\citenamefont
  {Fogarty}, \citenamefont {Veldhorst}, \citenamefont {Harper}, \citenamefont
  {Yang}, \citenamefont {Bartlett}, \citenamefont {Flammia},\ and\
  \citenamefont {Dzurak}}]{Fogarty2015}%
  \BibitemOpen
  \bibfield  {author} {\bibinfo {author} {\bibfnamefont {M.~A.}\ \bibnamefont
  {Fogarty}}, \bibinfo {author} {\bibfnamefont {M.}~\bibnamefont {Veldhorst}},
  \bibinfo {author} {\bibfnamefont {R.}~\bibnamefont {Harper}}, \bibinfo
  {author} {\bibfnamefont {C.~H.}\ \bibnamefont {Yang}}, \bibinfo {author}
  {\bibfnamefont {S.~D.}\ \bibnamefont {Bartlett}}, \bibinfo {author}
  {\bibfnamefont {Steven~T.}\ \bibnamefont {Flammia}}, \ and\ \bibinfo {author}
  {\bibfnamefont {A.~S.}\ \bibnamefont {Dzurak}},\ }\bibfield  {title}
  {\enquote {\bibinfo {title} {{Nonexponential fidelity decay in randomized
  benchmarking with low-frequency noise}},}\ }\href {\doibase
  10.1103/PhysRevA.92.022326} {\bibfield  {journal} {\bibinfo  {journal}
  {Physical Review A}\ }\textbf {\bibinfo {volume} {92}},\ \bibinfo {pages}
  {022326} (\bibinfo {year} {2015})},\ \Eprint
  {http://arxiv.org/abs/arXiv:1502.05119v2} {arXiv:arXiv:1502.05119v2}
  \BibitemShut {NoStop}%
\bibitem [{\citenamefont {Wood}\ and\ \citenamefont
  {Gambetta}(2018)}]{Wood2018}%
  \BibitemOpen
  \bibfield  {author} {\bibinfo {author} {\bibfnamefont {Christopher~J.}\
  \bibnamefont {Wood}}\ and\ \bibinfo {author} {\bibfnamefont {Jay~M.}\
  \bibnamefont {Gambetta}},\ }\bibfield  {title} {\enquote {\bibinfo {title}
  {Quantification and characterization of leakage errors},}\ }\href {\doibase
  10.1103/PhysRevA.97.032306} {\bibfield  {journal} {\bibinfo  {journal} {Phys.
  Rev. A}\ }\textbf {\bibinfo {volume} {97}},\ \bibinfo {pages} {032306}
  (\bibinfo {year} {2018})},\ \Eprint {http://arxiv.org/abs/1704.03081}
  {arXiv:1704.03081} \BibitemShut {NoStop}%
\bibitem [{\citenamefont {Flammia}\ and\ \citenamefont
  {Liu}(2011)}]{Flammia2011}%
  \BibitemOpen
  \bibfield  {author} {\bibinfo {author} {\bibfnamefont {Steven~T.}\
  \bibnamefont {Flammia}}\ and\ \bibinfo {author} {\bibfnamefont {Yi-Kai}\
  \bibnamefont {Liu}},\ }\bibfield  {title} {\enquote {\bibinfo {title} {Direct
  fidelity estimation from few {P}auli measurements},}\ }\href {\doibase
  10.1103/PhysRevLett.106.230501} {\bibfield  {journal} {\bibinfo  {journal}
  {Phys. Rev. Lett.}\ }\textbf {\bibinfo {volume} {106}},\ \bibinfo {pages}
  {230501} (\bibinfo {year} {2011})},\ \Eprint {http://arxiv.org/abs/1104.4695}
  {arXiv:1104.4695} \BibitemShut {NoStop}%
\bibitem [{\citenamefont {{Dirkse}}\ \emph {et~al.}(2018)\citenamefont
  {{Dirkse}}, \citenamefont {{Helsen}},\ and\ \citenamefont
  {{Wehner}}}]{Dirkse2018}%
  \BibitemOpen
  \bibfield  {author} {\bibinfo {author} {\bibfnamefont {B.}~\bibnamefont
  {{Dirkse}}}, \bibinfo {author} {\bibfnamefont {J.}~\bibnamefont {{Helsen}}},
  \ and\ \bibinfo {author} {\bibfnamefont {S.}~\bibnamefont {{Wehner}}},\
  }\bibfield  {title} {\enquote {\bibinfo {title} {{Efficient Unitarity
  Randomized Benchmarking of Few-qubit Clifford Gates}},}\ }\href@noop {}
  {\bibfield  {journal} {\bibinfo  {journal} {ArXiv e-prints}\ } (\bibinfo
  {year} {2018})},\ \Eprint {http://arxiv.org/abs/1808.00850} {arXiv:1808.00850
  [quant-ph]} \BibitemShut {NoStop}%
\bibitem [{\citenamefont {Katz}\ \emph {et~al.}(1978)\citenamefont {Katz},
  \citenamefont {Baptista}, \citenamefont {Azen},\ and\ \citenamefont
  {Pike}}]{Katz1978}%
  \BibitemOpen
  \bibfield  {author} {\bibinfo {author} {\bibfnamefont {D.}~\bibnamefont
  {Katz}}, \bibinfo {author} {\bibfnamefont {J.}~\bibnamefont {Baptista}},
  \bibinfo {author} {\bibfnamefont {S.~P.}\ \bibnamefont {Azen}}, \ and\
  \bibinfo {author} {\bibfnamefont {M.~C.}\ \bibnamefont {Pike}},\ }\bibfield
  {title} {\enquote {\bibinfo {title} {Obtaining confidence intervals for the
  risk ratio in cohort studies},}\ }\href {http://www.jstor.org/stable/2530610}
  {\bibfield  {journal} {\bibinfo  {journal} {Biometrics}\ }\textbf {\bibinfo
  {volume} {34}},\ \bibinfo {pages} {469} (\bibinfo {year} {1978})}\BibitemShut
  {NoStop}%
\bibitem [{\citenamefont {Epstein}\ \emph {et~al.}(2014)\citenamefont
  {Epstein}, \citenamefont {Cross}, \citenamefont {Magesan},\ and\
  \citenamefont {Gambetta}}]{Epstein2014}%
  \BibitemOpen
  \bibfield  {author} {\bibinfo {author} {\bibfnamefont {Jeffrey~M.}\
  \bibnamefont {Epstein}}, \bibinfo {author} {\bibfnamefont {Andrew~W.}\
  \bibnamefont {Cross}}, \bibinfo {author} {\bibfnamefont {Easwar}\
  \bibnamefont {Magesan}}, \ and\ \bibinfo {author} {\bibfnamefont {Jay~M.}\
  \bibnamefont {Gambetta}},\ }\bibfield  {title} {\enquote {\bibinfo {title}
  {{Investigating the limits of randomized benchmarking protocols}},}\ }\href
  {\doibase 10.1103/PhysRevA.89.062321} {\bibfield  {journal} {\bibinfo
  {journal} {Physical Review A}\ }\textbf {\bibinfo {volume} {89}},\ \bibinfo
  {pages} {062321} (\bibinfo {year} {2014})},\ \Eprint
  {http://arxiv.org/abs/arXiv:1308.2928} {arXiv:1308.2928} \BibitemShut
  {NoStop}%
\bibitem [{\citenamefont {Ball}\ \emph {et~al.}(2016)\citenamefont {Ball},
  \citenamefont {Stace}, \citenamefont {Flammia},\ and\ \citenamefont
  {Biercuk}}]{Ball2016}%
  \BibitemOpen
  \bibfield  {author} {\bibinfo {author} {\bibfnamefont {Harrison}\
  \bibnamefont {Ball}}, \bibinfo {author} {\bibfnamefont {Thomas~M.}\
  \bibnamefont {Stace}}, \bibinfo {author} {\bibfnamefont {Steven~T.}\
  \bibnamefont {Flammia}}, \ and\ \bibinfo {author} {\bibfnamefont
  {Michael~J.}\ \bibnamefont {Biercuk}},\ }\bibfield  {title} {\enquote
  {\bibinfo {title} {{Effect of noise correlations on randomized
  benchmarking}},}\ }\href {\doibase 10.1103/PhysRevA.93.022303} {\bibfield
  {journal} {\bibinfo  {journal} {Physical Review A}\ }\textbf {\bibinfo
  {volume} {93}},\ \bibinfo {pages} {022303} (\bibinfo {year} {2016})},\
  \Eprint {http://arxiv.org/abs/1504.05307} {arXiv:1504.05307} \BibitemShut
  {NoStop}%
\end{thebibliography}%

\end{document}